\documentclass[sigconf]{acmart}
\pdfoutput=1
\renewcommand\footnotetextcopyrightpermission[1]{} % removes footnote with conference information in first column
\usepackage[boxruled,vlined,linesnumbered]{algorithm2e}
\usepackage{graphicx}
\usepackage{wrapfig}
\usepackage{subcaption}
\def\argmin{\mathop{\rm argmin}}% argmin; indices appear below word

\usepackage{booktabs} % For formal tables
\usepackage{amsmath}
\usepackage{amssymb}
\usepackage{amsthm}

\setcopyright{none}
\begin{document}
\title{Optimal Net-Load Balancing in Smart Grids with High PV Penetration}
\titlenote{This work has been funded by the U.S. National Science Foundation under grant number ACI 1339756 and the Department of Energy (DoE) under award number DE-EE0008003.}

%\subtitle{Extended Abstract}
%\subtitlenote{The full version of the author's guide is available as
%  \texttt{acmart.pdf} document}

\author{Sanmukh R. Kuppannagari}
\affiliation{\institution{Ming Hsieh Department of Electrical Engineering, University of Southern California, Los Angeles, CA-90089}}
\email{kuppanna@usc.edu}

\author{Rajgopal Kannan}
\affiliation{\institution{US Army Research Lab, 12015 Waterfront Drive, Playa Vista, CA-90094}}
\email{rajgopal.kannan.civ@mail.mil}

\author{Viktor K. Prasanna}
\affiliation{\institution{Ming Hsieh Department of Electrical Engineering, University of Southern California, Los Angeles, CA-90089}}
\email{prasanna@usc.edu}

%\renewcommand{\shortauthors}{S. Kuppannagari et al.}

\begin{abstract}
Mitigating Supply-Demand mismatch is critical for smooth power grid operation. Traditionally, load curtailment techniques such as Demand Response (DR) have been used for this purpose. However, these cannot be the only component of a net-load balancing framework for Smart Grids with high PV penetration. These grids can sometimes exhibit supply surplus causing over-voltages. Supply curtailment techniques such as Volt-Var Optimizations are complex and computationally expensive. This increases the complexity of net-load balancing systems used by the grid operator and limits their scalability. Recently new technologies have been developed that enable the rapid and selective connection of PV modules of an installation to the grid. Taking advantage of these advancements, we develop a unified optimal net-load balancing framework which performs both load and solar curtailment. We show that when the available curtailment values are discrete, this problem is NP-hard and develop bounded approximation algorithms for minimizing the curtailment cost. Our algorithms produce fast solutions, given the tight timing constraints required for grid operation. We also incorporate the notion of fairness to ensure that curtailment is evenly distributed among all the nodes. Finally, we develop an online algorithm which performs net-load balancing using only data available for the current interval. Using both theoretical analysis and practical evaluations, we show that our net-load balancing algorithms provide solutions which are close to optimal in a small amount of time. 
\end{abstract}

\maketitle

\section{Introduction}
% a para about smart grid
Electrical power grids have undergone a drastic transformation since the 1970s in terms of both scale and complexity~\cite{ten2008vulnerability}. Technological advances such as the use of bi-directional AMI meters, allowing real time remote monitoring and control, have transformed them into smart grids~\cite{moslehi2010reliability}. 

% a para about solar power
Adoption of distributed solar energy has increased dramatically due to the falling cost of solar PVs. The installed prices of U.S. residential and commercial PV systems declined 5-7\% on average during 1998-2011~\cite{pvprices}. As per the DoE SunShot vision document, solar generated power is expected to grow to 14\% of the total power supply in 2030 and 27\% by 2050~\cite{sunshot}.

% uncertainties of solar power
Ensuring the matching of demand (load) and supply in a smart grid, also known as net-load balancing, is a critical grid operation. However, the increase in available power supply from solar energy is opening up new challenges in net-load balancing~\cite{basak2012literature}. Solar energy is heavily influenced by the ever changing weather conditions. This high variability in solar generation can lead to frequent demand-supply mismatches. This is especially pronounced in a distribution grid where a significant portion of supply comes from solar generation. This mismatch, if left unmitigated, can lead to 1) blackouts, if the demand is higher than the generation or 2) cause over-voltages and equipment tripping requiring manual intervention, if generation is higher than demand~\cite{gagrica2015microinverter}. 

% a para on current mitigation techniques
Load curtailment techniques for net-load balancing have been studied widely~\cite{albadi2008summary}. However, the issue of surplus supply must also be addressed to avoid over-voltages. Voltage Var Optimization (VVO) is a technique used to mitigate the over-voltage problem. VVO works by injecting the required amount of reactive power to reduce voltages to within the tolerable range~\cite{rahimi2012evaluation}. Grid operations have tight timing constraints and require solutions with low response time. Calculating the right amount of reactive power to be injected at each node of the grid requires solving Optimal Power Flow (OPF) equations, which are not scalable~\cite{rahimi2012evaluation}. New PV technology allows us to leverage the micro-inverters installed at PV installations. These micro-inverters provide the capability to (dis)connect a subset of PVs from each installation in the grid~\cite{gagrica2015microinverter}. We leverage this capability in our framework. For each PV installation, this provides us with a discrete set of solar curtailment strategies.

% our work, our contributions
In this work, we develop a net-load balancing framework which can perform both supply and demand curtailment over a horizon. Determining load or supply curtailment strategies when each strategy exhibits a discrete curtailment value is, as we show in this work, an NP-hard problem. Current techniques for curtailment strategy selection provide computationally expensive optimal solutions or faster heuristics with no optimality bounds. In contrast, we develop fast and optimal net-load balancing algorithms as a core component of our framework. Our algorithms minimize the cost of curtailment while ensuring that several practical constraints such as achieving the curtailment target, fairness etc. are met. We also develop an online heuristic to address scenarios where load and generation predictions for the entire horizon are not available beforehand. Using both theoretical analysis and practical evaluations, we show that the solutions provided by our net-load balancing algorithms are both scalable and near optimal.

\section{Related Works}
Significant literature exists on performing net-load balancing using load curtailment techniques such as Demand Response (DR). The key idea is to `shift' the loads away from high demand periods. Load curtailment can be pricing based in which the customers are incentivised or penalized to curtail. Works such as~\cite{chen2012optimal,roos1996modelling} fall into this category. Curtailment can also be performed using direct control from the grid operator. This scenario is better suited for micro-grids such as a University/Industrial campus.

Techniques which focus on direct control based load curtailment fall into two broad categories. The first category consists of stochastic optimization based approaches such as~\cite{kwac2013demand} and~\cite{chen2012real}. One limitation of such approaches is that they require a large number of nodes to ensure that the targeted curtailment is met with high probability -- this may not always be feasible~\cite{kuppannagari2016optimal}. Another approach is deterministic load curtailment in which nodes adopt curtailment strategies. In the real world (including our campus microgrid experience), strategies have discrete curtailment values. Nodes (buildings) can choose strategies from the strategy space such that the total curtailment objective is  satisfied while other practical constraints are met. The strategy selection problem, in general, is  NP-hard and hence it is difficult to get exact results in a reasonable amount of time. Here again there are two approaches. The first is to forgo accuracy guarantees in favor of performance. Techniques such as~\cite{roy2017},~\cite{zois2014efficient} and~\cite{ruiz2009direct} develop fast algorithms which can have arbitrarily large errors in the objective function (utility maximization, cost minimization etc.). Authors in~\cite{roy2017} develop a genetic algorithm based heuristic while \cite{zois2014efficient} presents a heuristic based on change making. The algorithm developed in~\cite{ruiz2009direct} uses Linear Programming whose solutions need to be rounded to integral values and can have large errors (unbounded integrality gap). The second approach is to provide computationally expensive exact solutions, for example,~\cite{luo1998milp} and~\cite{barbato2015energy}, where the authors use Mixed ILP for their algorithm. Previously, we developed polynomial time approximation algorithms for ``Sustainable'' Demand Response in which aggregate curtailment was bounded over intervals of the DR event~\cite{kuppannagari2016optimal,kuppannagari2015}. However, we did not consider net-load balancing along with fairness and curtailment cost objectives, as proposed.

Load Curtailment techniques are ineffective when supply due to solar PVs exceeds the demand. If this is left unmitigated, it causes over-voltages in the system leading to failures. Several works perform reactive solar curtailment in response to rising voltage. VVO~\cite{rahimi2012evaluation} increases reactive power to lower the voltage due to real power while iPlug curtails the solar energy input to the grid by redirecting it to charge storage or coordinate with local demand ramp-up resources~\cite{rongali2016iplug}. The authors in~\cite{lee2017distributed,tonkoski2011coordinated,singh2017sunshade} achieve continuous curtailment from solar PVs by running them at voltages other than the Maximum Power Point (MPP). This requires fine grained control of the solar panels. For some scenarios, fine grained control might not be available due to limitations of inverter technology. Our work addresses such scenarios through a curtailment model that handles a discrete set of curtailment values and provides bounded polynomial time approximations for achieving discrete curtailment targets.  As mentioned in~\cite{gagrica2015microinverter}, discrete solar curtailment can be performed by simply disconnecting individual PV modules using the micro-inverters installed at PV installations. As opposed to the technique developed in~\cite{gagrica2015microinverter}, which is reactive to over-voltages, we perform proactive solar curtailment.

%The assumption here is the presence of software controls to enable operating at these other voltages. The works mentioned above achieve continuous curtailment, the optimal values of which can be determined using polynomial time convex optimization algorithm. We focus on solar curtailment techniques which achieve discrete curtailment values, thus making the problem harder. As mentioned in~\cite{gagrica2015microinverter}, discrete solar curtailment can be performed by simply disconnecting individual PV modules using the micro-inverters installed PV installations. As opposed to the technique developed in~\cite{gagrica2015microinverter}, which is reactive to over-voltages, we perform proactive solar curtailment.  

%Works such as~\cite{moradi2017optimal} develop a scheduling algorithm which determines the schedule of generation from various sources for the entire day. The sources include a mix of renewable and non-renewable sources and the objective is to minimize cost. The resulting optimization problem is complex as it has to make a decision for each node for each interval of the entire day. This makes it difficult to scale this algorithm for large distribution grid or micro-grids with rooftop PVs installed. genetic~\cite{logenthiran2012demand} and

\section{Our Contributions}
Our work tries to address the limitations of the current frameworks by making the following contributions:
\begin{enumerate}
	\item To the best of our knowledge, ours is the first work to develop a curtailment strategy selection framework which can perform discrete solar curtailment pro actively to avoid over-voltages.
	\item We develop a unified framework which performs both load and solar curtailment. This greatly simplifies the overhead involved in grid management for the operator.
	\item We develop algorithms which are fast and provide worst case accuracy guarantee. Hence, we can simultaneously achieve the conflicting goals of accuracy and computational tractability with an ability to trade-off one for the other.
	\item We incorporate the notion of fairness into our algorithms and also develop an online algorithm for the cases when forecasts for the entire horizon are unavailable.
\end{enumerate}

\section{Net-Load Balancing in Smart Grids}
\subsection{Motivation}
The Smart Grid that we consider in this work consists of several demand nodes: consumers of electricity, and supply nodes: electricity producers. The supply nodes are the customers who have solar PVs installed. A node can act both as a demand node and a supply node. The Smart Grid has a high PV penetration i.e., the supply from solar PVs under normal weather conditions meet the demand of the consumers for most of the day. We assume that during night or during extremely unfavorable weather conditions, conventional sources of electricity are used to meet the demand. 

Mitigating supply-demand mismatch within tight timing constraints is critical for smooth operation of a smart grid. As shown in Figure~\ref{fig:fig1},  during several intervals of the day, such as regions 1 and 3, the demand of the consumers can exceed the solar supply. This can cause blackouts in the grid. Demand curtailment strategies need to be adopted during such intervals to avoid blackouts. The other extreme is shown using region 2 in Figure~\ref{fig:fig1}. These are the intervals in which the supply due to solar PVs exceeds the demand. This can cause over-voltages in the grid leading to the tripping of fault prevention devices~\cite{gagrica2015microinverter}. Under this scenario, supply curtailment strategies are required. 

We assume there exists a centralized grid operator with the capability of remotely switching a node into a curtailment strategy. A general framework which performs both load and supply curtailment greatly simplifies grid management for the grid operator.    

\begin{figure}[h!]
	\centering
	\includegraphics[width=0.99\linewidth]{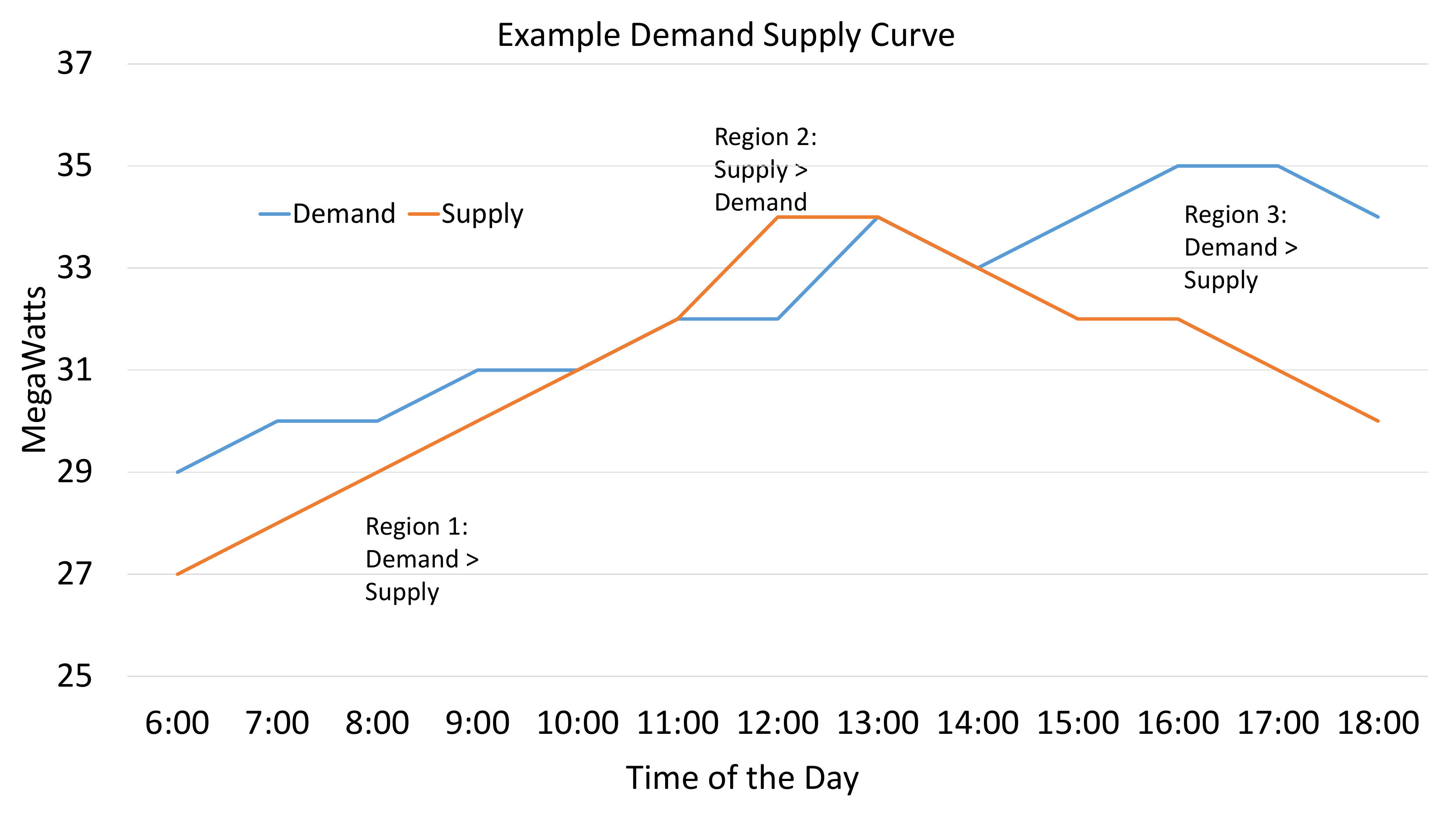}
%	\vspace{-0.5cm}
	\caption{Example Supply Demand Curve}
	\label{fig:fig1}
\end{figure}

\subsection{Demand Curtailment Strategies}
The Demand Curtailment model considered in this work is based on a real world Demand Response implementation in USC's Campus Microgrid~\cite{kuppannagariimplementation}. In the Smart Grid, each demand node is associated with several demand curtailment strategies. Examples of strategies include Global Zone Temperature Reset (GTR), Variable Frequency Drive Speed Reset (VFD), Equipment Duty Cycling (Duty) and their combinations~\cite{kuppannagariimplementation}. Each curtailment strategy for each node in a given time interval exhibits a discrete curtailment value. This complicates the problem as Linear Programming based techniques, which are both fast and optimal can no longer be used. This value can be predicted using algorithms mentioned in~\cite{aman2016learning}. Each node is also associated with a default curtailment strategy of curtailment value 0. Hence, if a node has no curtailment strategy available or does not participate in demand curtailment, we assume that it follows the default strategy. Demand curtailment is known in the literature as Demand Response~\cite{albadi2007demand}. 
     
\subsection{Solar Curtailment Strategies}
Each supply node i.e., a node with PV installation in the Smart Grid consists of several solar panels (each solar panel is called a module). Traditionally, modules are connected in series to an inverter which in turn is connected to the grid. However, this topology affects the efficiency of the PV system as the inverter conditions the output according to the poorest performing module~\cite{gagrica2015microinverter}. Therefore newer designs, in which each module is independently connected to a micro-inverter are becoming increasingly popular~\cite{deline2010partially}.

Technically, each micro-inverter of a PV installation is an independent grid connected generator, turning the PV installation into a segmented generator with discrete generation output. The maximum output of the PV installation will occur when all the solar panels are allowed to feed into the grid. However, at any given time, micro-inverters can be configured such that only a subset of PV modules are connected to the grid. Our objective is to exploit this capability by controlling the micro-inverter configuration and enabling discrete curtailment of supply. We refer the reader to~\cite{gagrica2015microinverter} for more details on utilizing micro-inverters for solar curtailment. Note that the technique developed in~\cite{gagrica2015microinverter} is a reactive technique which reacts to voltage increase and requires high frequency voltage sampling. Our technique is a proactive technique which avoids an increase in voltage by reducing supply in advance.       

%In Smart Grids with high PV penetration, certain periods can exhibit solar generation higher than the demand of the grid. If no action is taken under such scenarios, over-voltages will occur leading to inverter trips and thus no feed-in from the corresponding solar PVs~\cite{gagrica2015microinverter}. Volt-Var optimization is used to reduce over-voltages by increasing the reactive power. However, it requires inverters with largely overrated capacity and causes losses and deteriorates the power quality~\cite{gagrica2015microinverter}. Hence, we leverage micro-inverter based solar curtailment techniques~\cite{gagrica2015microinverter}.    

%Each supply node in the Smart Grid with a PV installation consists of one or more solar panels. Based on the type of the solar PVs and the weather conditions, each solar panel contributes to a certain amount of generation from the entire installation at any given time. Micro-inverters, at any given time, can be used to select which of the solar panels' output is fed into the Smart Grid and which of them are disconnected from the grid. The maximum output, for any given time, of the PV installation will occur when all the solar panels are allowed to feed into the grid. Hence, the PV installation can exhibit discrete solar curtailment values which can be determined by the available configurations of disconnecting the solar panels from the Smart Grid.  

\subsection{Curtailment Cost}
Each curtailment strategy for each node is associated with a cost value as curtailment leads to a loss in utility. These costs are determined by the grid operator to reflect the loss. Typically, the costs are some function of the curtailment value e.g., if a node, by following a strategy curtails $\gamma$, then the cost of this strategy will be $f(\gamma)$, where $f$ is some function determined by the grid operator. Linear and quadratic functions are commonly used cost functions in grid operations. The objective of our framework is to minimize cost while performing net-load balancing.

\begin{figure}[h!]
	\centering
	\includegraphics[width=0.99\linewidth]{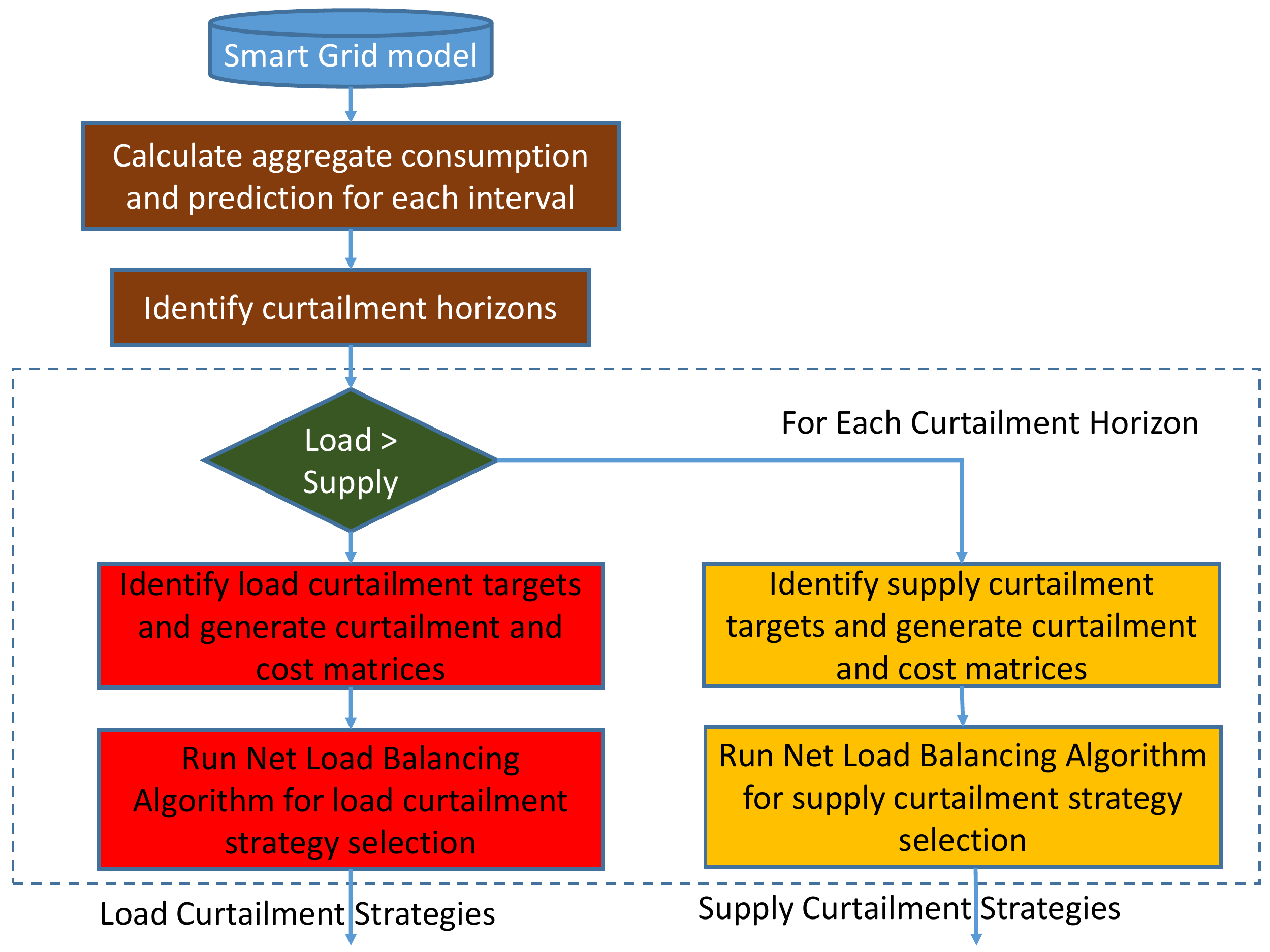}
%	\vspace{-0.4cm}
	\caption{High Level Overview of Net-Load Balancing Framework}
	\label{fig:flowchart}
\end{figure}

\subsection{Net-Load Balancing Framework}
In this work, we develop a generalized framework which performs net-load balancing by selecting load or supply curtailment strategies. We define net-load balancing horizon as the time horizon during which the net-load balancing framework is used. Net-load balancing horizon is divided into several smaller curtailment horizons. A curtailment horizon is defined as a period of time during which either demand is higher than the supply requiring a demand curtailment or vice-versa. 

\subsubsection{Model} In our net-load balancing framework, each supply node of the Smart Grid is associated with a generation prediction model such as ARIMA+ANN ensemble~\cite{ramsami2015hybrid}. Similarly, each demand node is associated with a demand prediction model such as ARIMA~\cite{aman2015prediction}. Each node is also associated with a curtailment prediction model~\cite{aman2016learning}. Determining the best prediction model for each node is a separate research topic and is out of the scope of this work. This work uses historical load, generation and curtailment prediction data.

\subsubsection{Method} The flow chart in Figure~\ref{fig:flowchart} gives a high level overview of our net-load balancing framework. The framework determines the aggregate load and supply for each interval in the net-load balancing horizon. It then identifies a list of load curtailment horizons and supply curtailment horizons and the respective curtailment targets (curtailment target calculation is discussed in Section~\ref{ssec:sgm}). For each interval of every curtailment horizon, the net-load balancing framework uses the curtailment prediction models to determine discrete curtailment values for each node. Using the discrete curtailment values, it also generates the cost values associated with them. Then, for each curtailment horizon, it runs one of the net-load balancing algorithms detailed in Section~\ref{sec:NLB}. The algorithm to run is pre-determined by the grid operator. The algorithm returns the curtailment strategies to be followed by each node in each time interval of each curtailment horizon. The model, objective and the algorithms for each curtailment horizon are formally discussed in Section~\ref{sec:NLB}.      

We realize that the prediction models incur a prediction error which creates uncertainties in net-load balancing solutions. In this work, we do not address these uncertainties. We plan to address these issues in our future work. 

\begin{table}
	\caption{List of Variables in the Models}
	%	\vspace{-0.4cm}
	\begin{tabular}{|l|p{6.5 cm}|}
		\hline
		Variable & Meaning \\ \hline
		$M$ & Number of nodes \\
		$N$ & Number of curtailment strategies \\
		$T$ & Number of time intervals in curtailment horizon \\
		$\gamma_{bj}(t)$ & Curtailment achieved by node $b$ following curtailment strategy $j$ at time $t$ \\
		$c_{bj}(t)$ & Cost of node $b$ following curtailment strategy $j$ at time $t$. Essentially, cost associated with $\gamma_{bj}(t)$ \\
		$x_{bj}(t)$ & 0-1 decision variable which denotes whether node $b$ should follow (1) strategy $j$ at time $t$ or not (0) \\
		$\Gamma_t$ & Curtailment target for interval $t$ \\
		$\Gamma$ & Upper bound on the curtailment achieved in the curtailment horizon \\
		$\alpha_b B_b, B_b$ & Lower bound and upper bound on the curtailment budget for node $b$ \\
		$\Theta_t$ & Dynamic programming recursion function (and table) used in Algorithm 1 \\
		$\Phi$ & Dynamic programming recursion function (and table) used in Algorithm 2 \\ 
		$\Gamma_{min}$ & Smallest curtailment target among all the intervals \\
		$ \mu $ & Constant used to round curtailment values in the approximation algorithm. Defined as $\frac{\epsilon \Gamma_{min}}{M}$\\
		$\widehat{\gamma}_{bj}(t)$ &  Rounded curtailment values. Defined as $\widehat{\gamma}_{bj}(t) = \lceil \frac{\gamma_{bj}(t)}{\mu} \rceil$ \\
		\hline
	\end{tabular}
	
	\label{variables}
\end{table}

\section{Net-Load Balancing Algorithms}
\label{sec:NLB}
\subsection{SmartGrid Model for Curtailment Selection}
\label{ssec:sgm}
As per our model, the Smart Grid consists of $M$ nodes. For each node, there are $N$ curtailment strategies available. Let $T$ be the number of time intervals in the curtailment horizon i.e., the intervals during which we schedule the curtailment. We are given a time varying curtailment matrix $\gamma(t) \in R^{M \times N}$ with element $\gamma_{bj}(t)$ denoting the discrete curtailment obtained by node $b$ following curtailment strategy $j$ at time $t \in \{1,\dots,T\}$. For each time $t$, we are also given a cost matrix $C(t) \in R^{M \times N}$ where $c_{bj}(t)$ denotes the cost associated with node $b$ following curtailment strategy $j$. Let $X(t)$ be the decision matrix. An element $x_{bj}(t) = 1$ if node $b$ follows curtailment strategy $j$ in interval $t$ and 0 otherwise. 
For each interval $t$, we are given a curtailment target $\Gamma_t$ by the net-load balancing framework calculated by taking the difference between the aggregate supply and demand. $\Gamma_t$ represents the desirable curtailment target for each period, however, it might be exceeded. In order to limit wasteful curtailment, we are also given $\Gamma$, which denotes the upper bound on the achieved curtailment in the curtailment horizon. The notations used in the following sections are summarized in Table~\ref{variables}.

\subsection{Minimum Cost Net-Load Balancing}
Given the Smart Grid Model above, the objective of the Minimum Cost Net-Load Balancing Algorithm is to determine node-curtailment strategy pairs for each interval of the curtailment horizon such that: (1) The curtailment target $\Gamma_t$ for each interval $t$ is achieved, (2) the cost is minimized for the entire curtailment horizon, and (3) the aggregate curtailment across the entire curtailment horizon is no more than $\Gamma$. 

This problem, as we show in Section~\ref{sec:nphard} in the appendix is NP-hard. We first formulate the problem using an Integer Linear Program (ILP). However, the time complexity for solving ILPs is exponential. Hence, we develop a polynomial time approximation algorithm for the same.

More formally, we develop an algorithm with a runtime which is polynomial in the input size $M, N$ and $T$ and $\frac{1}{\epsilon}$, where $\epsilon$ is an approximation guarantee (accuracy) parameter, which ensures that objective of the problem is minimized and in the worst case the constraints are violated by a maximum factor of $(1\pm\epsilon)$. To develop the approximation algorithm, we use a dynamic programming algorithm (Equation~\ref{eqn:dp}) which for each interval $t$, determines the cost of achieving various curtailment values, each of which is $\geq \Gamma_t$. We then use another dynamic programming algorithm (Equation~\ref{eqn:dp2}) to combine the results of each interval to achieve an aggregated curtailment value of $\leq \Gamma$ with minimum cost. The sizes of the tables of both the dynamic programming algorithms are proportional to the maximum possible cost. This leads to very large runtime to solve the problem exactly. Hence, we scale and round the costs. The scaling and rounding causes several curtailment values to be indistinguishable, this introduces error into the value of our solution. Thus the resulting algorithm is an approximation algorithm (as opposed to exact algorithm) which produces an approximate solution in polynomial time which is independent of the maximum cost. We provide the worst case approximation guarantee for the algorithms.  

\subsubsection{ILP Formulation}
The ILP formulation for the Minimum Cost Net-Load Balancing problem is as follows:

\begin{flalign}
Minimize: & \sum_{b=1}^{M}\sum_{j=1}^{N}\sum_{t=1}^{T} c_{bj}(t) x_{bj}(t)& \label{eqn:mcnlbobj}\\
s.t. & \sum_{b=1}^{M}\sum_{j=1}^{N} \gamma_{bj}(t) x_{bj}(t) \geq \Gamma_t & \forall t \label{eqn:achievegamma}\\
& \sum_{b=1}^{M}\sum_{j=1}^{N}\sum_{t=1}^{T} \gamma_{bj}(t) x_{bj}(t) \leq \Gamma &  \label{eqn:gammaless}\\
& \sum_{j=1}^{N}x_{bj}(t) == 1 & \forall b,t \label{eqn:onestrat}\\
& x_{bj}(t) \in \{0,1\}& \forall b, j, t \label{eqn:int}
\end{flalign}

Equation~\ref{eqn:achievegamma} ensures that the curtailment target for each time interval is achieved. Equation~\ref{eqn:gammaless} ensures that the aggregate curtailment is less than the maximum limit $\Gamma$. Equation~\ref{eqn:onestrat} ensures that each node in each time interval follows exactly one strategy (possibly the default strategy with 0 curtailment value).

\subsubsection{Approximation Algorithm}
Let $\Gamma_{min} = \min_t\{\Gamma_t\}$ be the smallest curtailment target among all the intervals. Define $\mu = \frac{\epsilon \Gamma_{min}}{M}$. For each $\gamma_{bj}(t)$, define $\widehat{\gamma}_{bj}(t) = \lceil  \frac{\gamma_{bj}(t)}{\mu}\rceil$. Similarly, define $\widehat{\Gamma}$ and $\widehat{\Gamma_t} \forall t$. We refer to these values as rounded curtailment values.

For each interval $t$, we define a function $\Theta_t: \mathcal{Z}^+\cup\{0\} \times  \{1,\dots,M\}$. $\Theta_t(\widehat{\gamma}, b)$  denotes the minimum cost required to achieve a curtailment value of $\widehat{\gamma}$ using nodes  $1,\dots,b$ where $b \in \{1,\dots,M\}$. $\Theta_t$ can be defined recursively as:

\begin{flalign}
\Theta_t(\widehat{\gamma}, b) = \begin{cases}
\min_j\{c_{bj}(t)\}& \text{if } b = 1 \text{ and }\\ 
&\exists j\;|\; \widehat{\gamma} = \widehat{\gamma}_{bj}(t) \\
\infty & \text{if } b = 1 \text{ and } \\
&\widehat{\gamma} != \widehat{\gamma}_{bj}(t)\;\forall j \\
\infty & \text{if } \widehat{\gamma} < 0 \\
\min_j\{ \Theta_t(\widehat{\gamma} - \widehat{\gamma}_{bj}(t),  b-1)  + \widehat{\gamma}_{bj}(t)  \} & \text{otherwise} 
\end{cases} \label{eqn:dp}
\end{flalign}

\begin{algorithm}
	\label{algorithm:algo1}
	\caption{Determine node strategy pairs given $(c, \widehat{\gamma}) \in S_t$}
	\KwIn{$(c, \widehat{\gamma})$, $t$}
	$x_{bj}(t) \in X(t) \gets 0 \forall b, j$ \\
	
	$\gamma_{cur} \gets \widehat{\gamma}$ \\
	\For{$b = M \text{ to } 1$} 
	{ \label{line:for}
		\eIf{$b \neq 1$}
		{
			$j \gets \argmin_j\{\Theta(\gamma_{cur} - \widehat{\gamma}_{bj}(t), b-1) + c_{bj}(t) \}$	\label{line:argmaxj}\\
		}
		{
			$j \gets j\; |\; \widehat{\gamma}_{bj}(t) == \gamma_{cur}$	
		}
		$x_{bj}(t) \gets 1$ \\
		$\gamma_{cur} \gets \gamma_{cur} - \widehat{\gamma}_{bj}(t)$ \\
	}			
	
	\KwOut{Output $X(t)$, the list of curtailment strategies to be followed by each node in interval $t$.}
\end{algorithm}

\begin{algorithm}[!ht]
	\label{algo:Algo2}
	\caption{Minimum Cost Net-Load Balancing Algorithm}
	\KwIn{$C(t), \gamma(t), \Gamma_t, \forall t, \Gamma$}
	Compute $\widehat{\gamma}_{bj}(t), \widehat{\Gamma}_t \forall t, \widehat{\Gamma}$ \\
	Fill the table $\Theta_t \forall t$ using equation~\ref{eqn:dp} \label{line:filltheta}\\
	Compute $S_t = \{(\Theta_t (\widehat{\gamma},M), \widehat{\gamma})| \widehat{\gamma} \geq \widehat{\Gamma_t}\} \forall t$ \label{line:computest}\\
	Fill the table $\Phi$ using equation~\ref{eqn:dp2} \\
	$\widehat{\gamma}_{cur} \gets \argmin_{\widehat{\gamma}} \{ \Phi(\widehat{\gamma}, T)  |  \widehat{\gamma} \leq \widehat{\Gamma} \}$ \label{line:gammaless} \\
	\If{$\widehat{\gamma}_{cur} == \phi$}{
		No curtailment strategies exist \\	
		Exit Algorithm \\
	}
	\For{$t=T$ to 1}{
		\eIf{$t \neq 1$}
		{
			$j \gets \argmin_{j\;|\;(c_j, \widehat{\gamma}_j) \in S_t}\;\; \{\Phi(\widehat{\gamma}_{cur} - \widehat{\gamma}_j, t-1) 
			+ c_j \}$	\\
		}
		{
			$j \gets j\;|\; \widehat{\gamma} = \widehat{\gamma}_{j}, (c_j, \widehat{\gamma}_j) \in S_t$	\\
		}	
		Run Algorithm 1 with $(c_j,\widehat{\gamma}_j)$ to get $X(t)$ \\
		$\widehat{\gamma}_{cur} \gets \widehat{\gamma}_{cur} - \widehat{\gamma}_j$ \\
	}
	\KwOut{$X(t) \forall t$, the list of curtailment strategies to be followed by each node in each interval}
\end{algorithm}

The dynamic program can be solved by creating a table of size $k \times M$, where $k = \widehat{\Gamma}$  for each interval. For notational simplicity, we refer to table using the same variable $\Theta_t$ as the recursive function. Once the table is filled, for each interval $t$, we define a set $S_t = \{(\Theta_t (\widehat{\gamma},M), \widehat{\gamma})| \widehat{\gamma} \geq \widehat{\Gamma_t}\}$. For any element $(\Theta_t (\widehat{\gamma},M), \widehat{\gamma}) \in S_t$, Algorithm 1 can be used to determine the strategies to be followed by each node to achieve $\widehat{\gamma}$ with cost $\Theta_t (\widehat{\gamma},M)$ in the interval $t$.

Now, given $S_t \forall t \in \{1,\dots,T\}$, we need to select exactly one element $e_t=(c_t,\widehat{\gamma_t}) \in S_t \forall t$ such that $\sum_{t=1}^{T}\widehat{\gamma_t} \leq \widehat{\Gamma}$ and $\sum_{t=1}^{T}c_t$ is minimized. We define a function $\Phi: \mathcal{Z}^+\cup\{0\} \times \{1,\dots,T\}$. $\Phi(\widehat{\gamma}, t)$ denotes the minimum cost required to achieve the curtailment value of $\widehat{\gamma}$ and time intervals $1,\dots,t$. $\Phi$ can be defined recursively as: 

\begin{flalign}
\Phi(\widehat{\gamma}, t) = \begin{cases}
\min_j\{c_j\} & \text{if } t = 1 \text{ and } \\
& \exists j\;|\;\widehat{\gamma} = \widehat{\gamma}_{j}, (c_j, \widehat{\gamma}_j) \in S_t \\
\infty & \text{if } t = 1 \text{ and }\\
& \widehat{\gamma} != \widehat{\gamma}_{j}\; \forall j\;|\;(c_j, \widehat{\gamma}_j) \in S_t \\
\infty & \text{if } \widehat{\gamma} < 0 \\
\min_{j\;|\;(c_j, \widehat{\gamma}_j) \in S_t} \;\; \{\Phi(\widehat{\gamma} - \widehat{\gamma}_j, t-1) \;\;\;+ \widehat{\gamma}_j \}  & \text{otherwise}
\end{cases} \label{eqn:dp2}
\end{flalign}

This dynamic program requires a table of size $k \times T$ , where $k = \widehat{\Gamma}$. Again, for notational simplicity, we refer to table using the same variable $\Phi$ as the recursive function. Algorithm 2 can be used to determine the curtailment achieved in each interval and the corresponding node strategy pairs.

\begin{theorem}	\label{th:FPTA1}
	Algorithm 2 is a polynomial time algorithm for minimum cost net-load balancing which in the worst case violates the maximum curtailment constraint (Equation~\ref{eqn:gammaless}) by at most $(1+\epsilon)$ factor and violates the per interval curtailment target constraint (Equation~\ref{eqn:achievegamma}) by at most $(1 - \epsilon)$ factor. 
\end{theorem}

The proof of this theorem is discussed in the appendix in Section~\ref{sec:proofFPTA1}.

\subsection{Minimum Cost Net-Load Balancing with Fairness}
%Motivation for fairness and description of the problem. ILP formulation. problems with ILP. High level description of the approximation algorithm. Actual Algo. Theorem. Move proofs to appendix.

Curtailment from a node leads to a loss of utility for the node. Hence, it would be unfair to force some nodes to incur losses due to high curtailment while leaving others with minimal curtailment. The Minimum Cost Net Load Balancing Algorithm discussed in the previous section does not take fairness into account and can lead to solutions with uneven curtailment values from the nodes. We address the issue of fairness in this section by assigning a curtailment budget range (which can be set by the grid operator) to each node. The algorithm, by ensuring that no node incurs a curtailment more or less than its budgeted range over each curtailment horizon, ensures that net-load balancing is done in a fair manner.

Similar to the previous problem, we first develop an ILP formulation for this problem. We then relax the ILP into a Linear Program (which can be solved in polynomial time) and round back the results to integers. This rounding, however, violates certain constraints and increases the objective value. We provide guarantees on the worst case violation of the constraints and the increase in the objective value in the worst case. 

\subsubsection{ILP Formulation} 
\label{sssec:fairilp}
Let $[\alpha_b B_b, B_b]$ be the curtailment budget for node $b$ with $\alpha_b \in [0,1]$. Both $B_b$ and $\alpha_b$ are determined by the grid operator. The problem of minimum cost net-load balancing with fairness can be formulated using the following ILP: 
\begin{flalign}
Minimize: & \sum_{b=1}^{M}\sum_{j=1}^{N}\sum_{t=1}^{T} c_{bj}(t) x_{bj}(t)& \label{eqn:mcnlfbobj}\\
s.t. & \sum_{b=1}^{M}\sum_{j=1}^{N} \gamma_{bj}(t) x_{bj}(t) \geq \Gamma_t & \forall t \label{eqn:mcnlfachievegamma}\\
& \sum_{b=1}^{M}\sum_{j=1}^{N}\sum_{t=1}^{T} \gamma_{bj}(t) x_{bj}(t) \leq \Gamma &  \label{eqn:mcnlfgammaless}\\
& \alpha_b B_b \leq \sum_{j=1}^{N}\sum_{t=1}^{T} \gamma_{bj}(t) x_{bj}(t) \leq B_b& \forall b \label{eqn:mcnlffair}\\
& \sum_{j=1}^{N}x_{bj}(t) == 1 & \forall b,t \label{eqn:mcnlfonestrat}\\
& x_{bj}(t) \in \{0,1\}& \forall b, j, t \label{eqn:mcnlfint}
\end{flalign}

Equation~\ref{eqn:mcnlffair} is the additional constraint added that ensures that each node curtails an amount within its budgeted range. 

\subsubsection{Approximation Algorithm}
The ILP formulated in Section~\ref{sssec:fairilp}, when relaxed to a Linear Program will lead to unbounded errors. Hence,  to develop an approximation algorithm with theoretical worst case bounds, we first make the following assumption: The costs $c_{bj}(t)$ are a function of $\gamma_{bj}(t)$ i.e., $c_{bj}(t) = f(\gamma_{bj}(t))$. We will derive approximation guarantees when the function $f$ is linear and when it is quadratic.

In order to develop a bounded approximation algorithm, we first relax the ILP to a linear program i.e., we replace Equation~\ref{eqn:mcnlfint} with $0 \leq x_{bj}(t) \leq 1\;\forall b,j,t$ and solve the Linear Program. Solving a linear program takes polynomial amount of time using methods such as inter-point or ellipsoid~\cite{boyd2004convex}. However, the solution will contain fractional values for the decision variables $x_{bj}(t),\;\forall b, j, t$ which need to be rounded to 0 or 1 for a valid solution. Now, naively rounding the decision variables leads to errors which are unbounded. Hence, we develop Algorithm 3 which is a novel rounding algorithm which guarantees that the constraints are violated by at most a factor of 2 in the worst case. For each $b, t$, the algorithm works by computing expected curtailment $\gamma' = \sum_{j=1}^{N}\gamma_{bj}(t)x_{bj}(t)$ and rounding it to the curtailment value $\gamma_{bj}(t)$ nearest to it. We have the following two results for this algorithm. The proofs are discussed in the appendix.

\begin{theorem}
	For a linear cost function $f$, Algorithm 3 is a (2,2)-factor Minimum Cost Net Load Balancing with Fairness Algorithm. The cost of the solution achieved by Algorithm 3 is at most twice the optimal while the budget and targeted curtailment constraints (Eqs~\ref{eqn:mcnlfgammaless} and~\ref{eqn:mcnlffair}) are violated by at most a factor of two. 
\end{theorem}   

\begin{theorem}
	For a quadratic cost function $f$, Algorithm 3 is (4,2)-factor algorithm.
\end{theorem} 

Note that the above guarantees are worst case guarantees. We discuss in the appendix the conditions under which these worst case guarantees occur. Knowing the worst case conditions and performance leaves the grid operator prepared for such scenarios. In practice, the performance is significantly better as shown using the experimental results.  

\begin{algorithm}
	\label{algo3}
	\caption{Minimum Cost Net-Load Balancing with Fairness}
	\KwIn{$C(t), \gamma(t), \Gamma_t, \forall t, \Gamma, B_b, \alpha_b \forall b$}
	$x_{bj}(t) \in X(t) \gets 0 \forall b, j, t$ \\
	Relax the ILP to an LP by replacing Equation~\ref{eqn:mcnlfint} with $0 \leq x_{bj}(t) \leq 1\;\forall b,j,t$ \\
	Solve LP to obtain solution $x^*_{bj}(t) \forall b, j, t$ \\
	\ForEach{$b,t$}{
		$\gamma' \gets \sum_{j=1}^{N}\gamma_{bj}(t)x^*_{bj}(t)$ \\
		Let $\gamma_{bi}(t) \leq \gamma' \leq \gamma_{bi+1}(t)$	\\
		\eIf{$(\gamma' - \gamma_{bi}(t)) \geq (\gamma_{bi+1}(t) - \gamma')$}{
		$x_{bi+1}(t) \gets 1$	
		}{
		$x_{bi}(t) \gets 1$
		}
	}
\KwOut{$X(t) \forall t$, the list of curtailment strategies to be followed by each node in each interval}
\end{algorithm}

\subsection{Online Algorithm for Fair Net-Load Balancing}
%Motivation for online algo. Actual Algo.
The algorithms discussed in the previous two sections require the availability of the curtailment prediction data for the entire horizon. However, certain scenarios require net-load balancing in an online manner. At the beginning of each interval, the data is made available and net-load balancing needs to be performed in a myopic way.

We develop a greedy online heuristic algorithm (Algorithm 4) for this problem. The algorithm finds a minimum cost way to achieve the curtailment target $\Gamma^o_l = \Gamma_t$ for the current interval while ensuring that no node curtails more than the budget for the current interval. The upper limit of the budget for the current interval $B^o_{b}$ for each node $b$ is determined by multiplying the ratio of $\frac{B_b}{\sum_{t=1}^{T}\Gamma_t}$ with $\Gamma^o_l$. The lower bound is simply $\alpha_b B_b$. The upper bound on curtailment $\Gamma^o_u$ is determined by multiplying the ratio of $\frac{\Gamma}{\sum_{t=1}^{T}\Gamma_t}$ with $\Gamma^o_l$. Let $\gamma^o$ denote the curtailment matrix and $C^o$ denote the cost matrix with $\gamma^o_{bj}$ being the curtailment for node $b$ following strategy $j$ and $c^o_{bj}$ being its cost. We remove any values $\gamma^o_{bj}$ which are outside the curtailment budget range from the curtailment matrix $\gamma^o$. Note that the values $\Gamma, \Gamma_t, B_b \forall b$ are unknown for the current net-load balancing horizon. They can be obtained from some past horizon.

\begin{algorithm}
	\caption{Online Algorithm for Fair Net-Load Balancing}
	\KwIn{$\Gamma^o_l, \gamma^o, C^o, \Gamma, \Gamma_t \forall t,  B_b, \alpha_b \forall b$}
%	\KwIn{$\Gamma^o_l, \Gamma^o_u, \gamma^o, C^o, B^o_b \forall b$}
	$\Gamma^o_u \gets \frac{\Gamma}{\sum_{t=1}^{T}\Gamma_t} \Gamma^o_l$ \\
	$B^o_b \gets \frac{B_b}{\sum_{t=1}^{T}\Gamma_t} \Gamma^o_l\;\forall b$ \\
	Compute $C^o_{min}, C^o_{max}, \widehat{c}^o_{bj}, \widehat{C}^o_{max}$ similar to Algorithm 2 using $\gamma^o$ in which curtailment values outside curtailment budget range are removed. \\
	Fill table $\Theta$ using equation 6 \\
	$c^* \gets min_{\widehat{c}}\{\widehat{c}| \Gamma^o_l \leq \Theta(\widehat{c},M) \leq \Gamma^o_u\}$ \\
	Run Algorithm 1 with $(c^*,\Theta(c^*,M))$ to get $X^o$ \\
	\KwOut{$X^o$, the list of curtailment strategies}
\end{algorithm}

\section{Results and Analysis}
In addition to providing theoretical guarantees, we perform practical evaluations of the algorithms developed in Section~\ref{sec:NLB}. We implemented the algorithms using MATLAB~\cite{matlab}. The LP and ILP algorithms were implemented using IBM ILOG Cplex Optimization Studio~\cite{cplex}. The experiments were performed on Dell optiplex with 4-cores and 4 GB RAM. A net-load balancing horizon of 32 15-min intervals was considered with 16 intervals of load curtailment and 16 intervals of solar curtailment. 

We evaluated the algorithms by varying the (L,U) pair where, for notational simplicity, L=$\sum_{t=1}^{T}\Gamma_t$ and U=$\Gamma$ (Section~\ref{sssec:fairilp}). Note that L and U are the lower and upper bound on the curtailment to be achieved in the curtailment horizon and hence represents the feasible curtailment range. The costs of the strategies were evaluated using the function $f(\gamma) = 2\gamma^2$, where $\gamma$ is the curtailment value of the strategy. Section~\ref{ssec:dataset} describes the input dataset generation. As the dataset was generated from historical data, perfect knowledge of the future was assumed with no prediction errors. 

\subsection{Dataset}
\label{ssec:dataset}
We obtained the load curtailment data from the demand response implementation on our University Campus. Our campus consists of 150 DR enabled nodes (buildings) each of which can follow 6 load curtailment strategies. The load curtailment values for each node-strategy pair was generated using algorithms mentioned in~\cite{aman2016learning}. We varied the load curtailment target from 500 to 1500 kWh.

Unlike load curtailment data for which we had a real world dataset, we had to simulate solar curtailment data. The output of a solar PV is determined mainly by the solar radiance at the PV installation, the PV area and the PV yield~\cite{pvout}. We used the hourly solar radiance data available at~\cite{solarrad} for the Los Angeles County. We then used the PV-output calculator available at~\cite{pvout} to calculate the solar generation data by varying the PV area from $10 m^2$ to $20 m^2$. We also varied the solar panel yield from $5\%$ to $15\%$. Hence, a fixed PV area and yield represents a node in our dataset. To obtain solar curtailment values, if the PV output for a given hour for a node was $O$, we generated 6 curtailment values: 0, $0.125*O, 0.25*O, 0.5*O, 0.75*O, O$. Hence, each curtailment value represents a PV connection/disconnection setting. All the 4 15-min intervals for a given hour were assigned the same curtailment values. For clarity, given an $(L,U)$ pair, we report the error of either solar or load curtailment, whichever one performs worse.

\subsection{Minimum Cost Net-Load Balancing}
We evaluate our Minimum Cost Net Load Balancing Algorithm (Algorithm 2) by varying the (L,U) pairs as discussed above. We perform experiments to compare the cost of the solution obtained by our algorithm against the optimal solution obtained by the ILP. We also perform a scalability analysis of our algorithm. 

\begin{figure}[ht]
	\begin{subfigure}{0.48\textwidth}
		\centering
		\includegraphics[width=\linewidth]{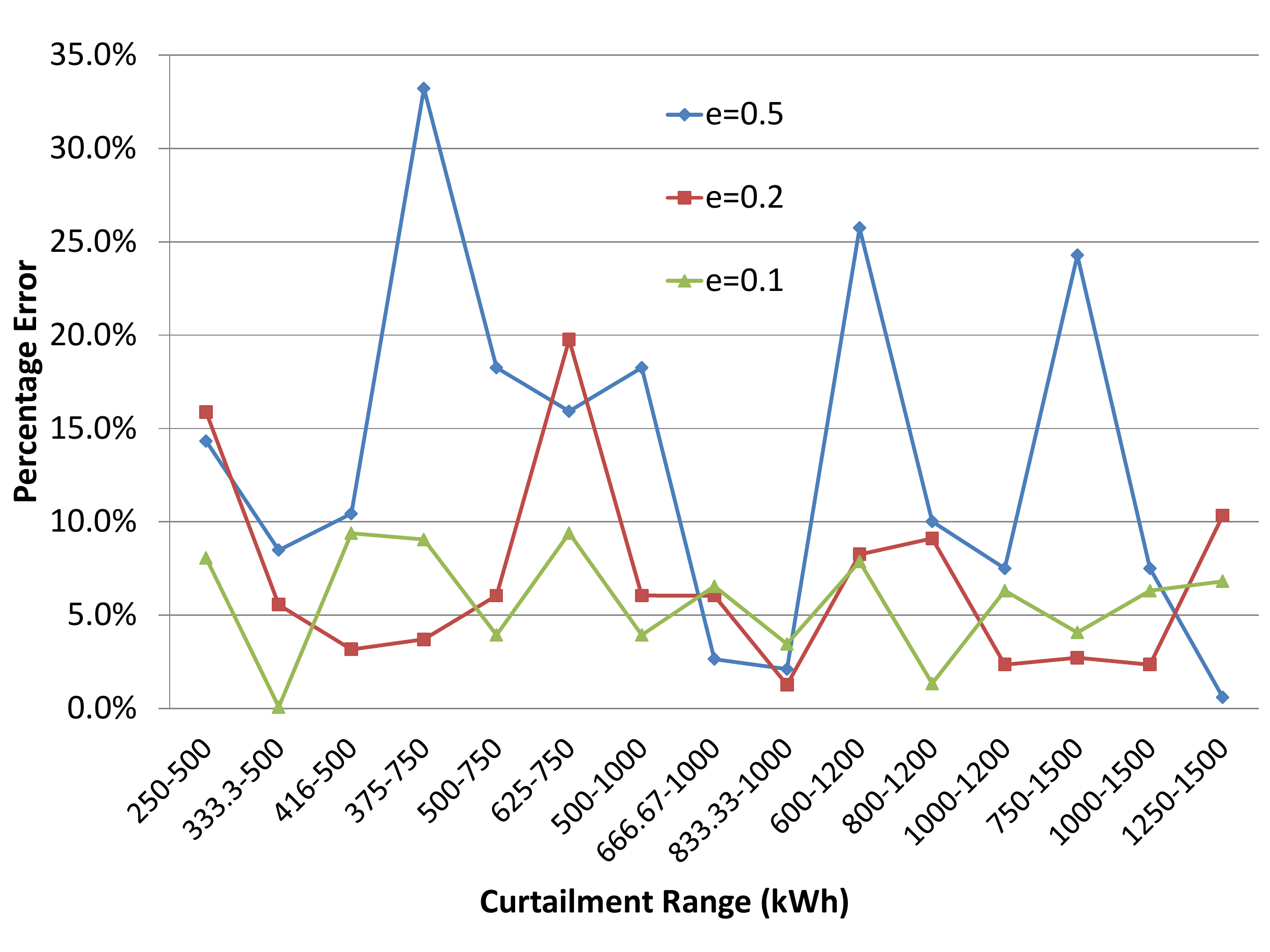}
		\caption{$\epsilon$ = 0.5-0.1 (50-10\% Error Guarantee)}
		\label{fig:accuracy-5-1}
	\end{subfigure}
	\begin{subfigure}{0.48\textwidth}
		\centering
		\includegraphics[width=\linewidth]{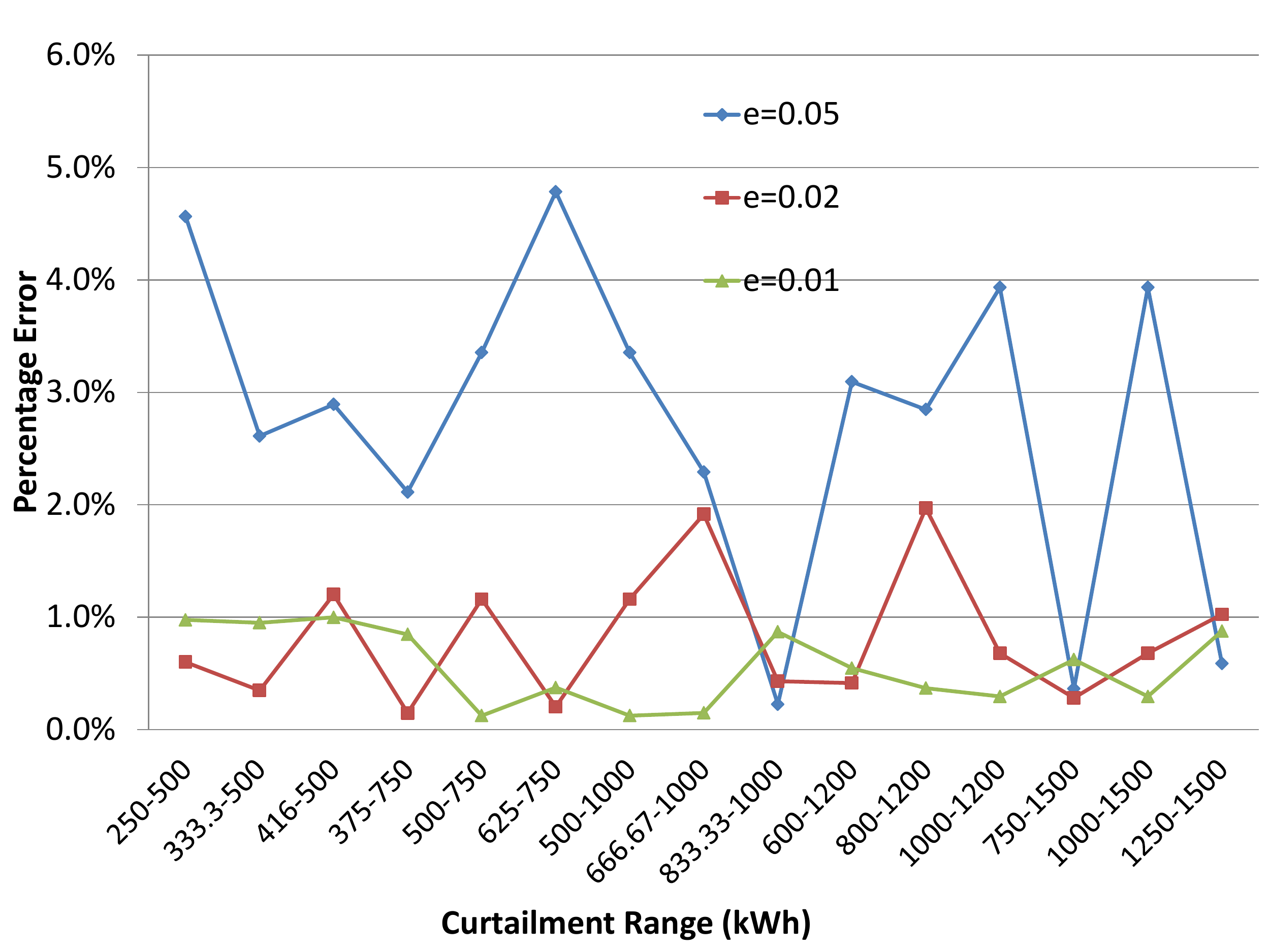}
		\caption{$\epsilon$ = 0.05-0.02 (5-2\% Error Guarantee)}
		\label{fig:accuracy-05-01}
	\end{subfigure}
	\caption{Percentage Error of the curtailment target constraints for various values of approximation factor $\epsilon$}
\end{figure}

\subsubsection{Accuracy Analysis} 
Figures~\ref{fig:accuracy-5-1} and~\ref{fig:accuracy-05-01} show the percentage error of the curtailment target constraint (Equation~\ref{eqn:achievegamma}) violation for various values of the theoretical guarantee $\epsilon$.  For example, if $\epsilon=0.05$, the algorithm will incur an error of 5\% in the worst case.  As we can note from the figures, the errors incurred by our solution are within the theoretical guarantees provided by the number $\epsilon$. In practice, the errors are much lower. For $\epsilon=0.5$ (50\%), the highest error incurred is 40\%. Similarly, for $\epsilon=0.2$  (20\%), barring a few cases, the errors are less than 15\%.

\begin{figure}
	\centering
	\includegraphics[width=0.6\linewidth]{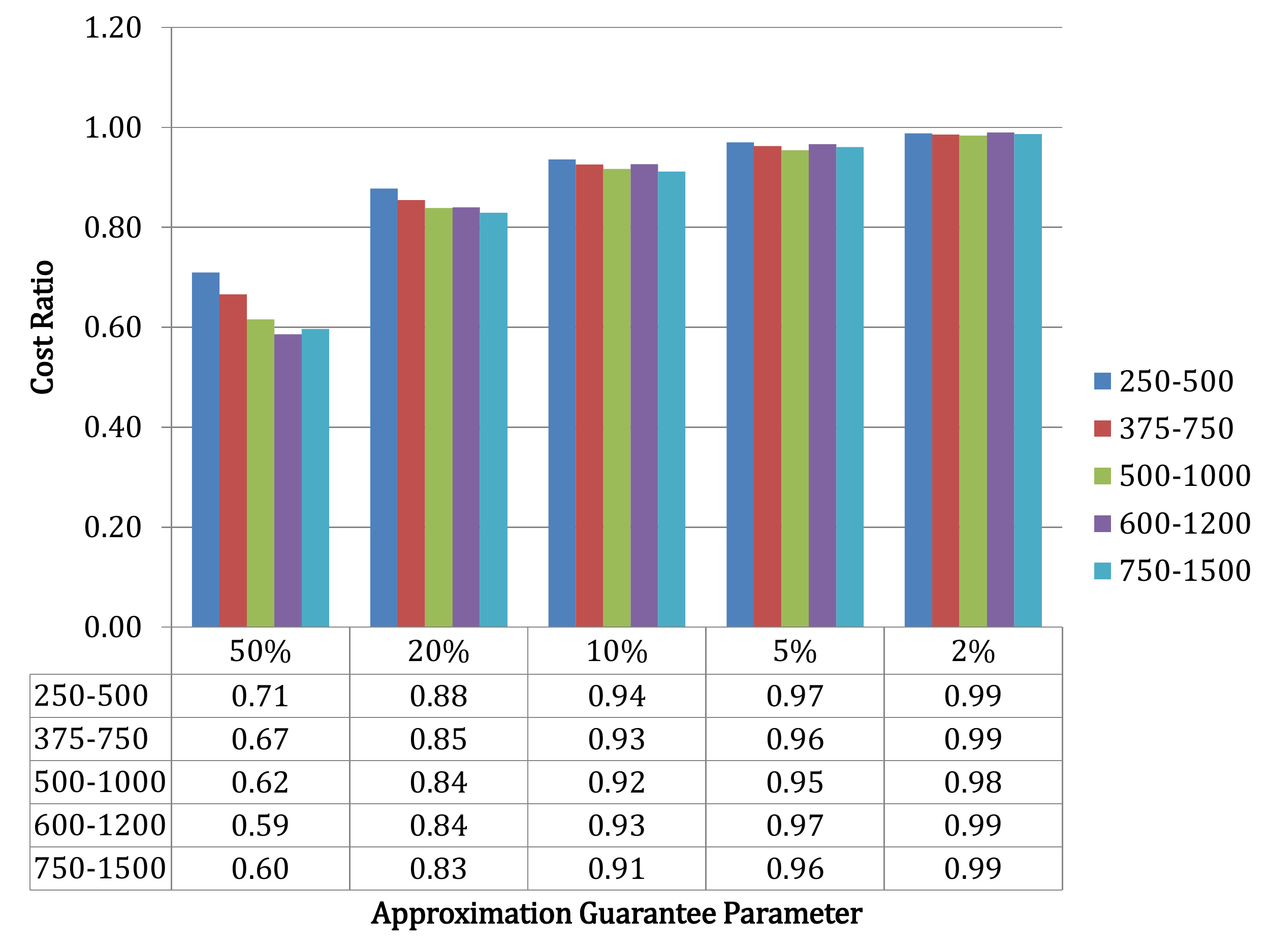}
	%\vspace{-0.4cm}
	\caption{Ratio of the cost of solution obtained by Algorithm 2 and the optimal solution versus the approximation guarantee ($\epsilon$)}
	\label{fig:costratio}
\end{figure}

As evident from Theorem 1, the cost of the solutions obtained from Algorithm 2 are less than or equal to the optimal solution. This is possible because instead of tightly satisfying the constraint, as the optimal solution does, Algorithm 2 tries to find a solution with a lower cost which possibly violates the constraints by a maximum of $\epsilon$. Figure~\ref{fig:costratio} shows the ratio of the cost of solution obtained by Algorithm 2 and the optimal solution versus the approximation guarantee ($\epsilon$). As we can notice from the figure, for each curtailment range, the ratio increases as the approximation guarantee is tightened i.e. reduced. For lower values of $\epsilon$ such as 0.02 (2\%), the ratio is close to 1. The ratio is never greater than 1 implying that the cost of the solution obtained by Algorithm 2 is always less than the optimal cost.

We can also note that for a fixed approximation guarantee ($\epsilon$), the ratio typically decreases with an increase in the upper bound U of the curtailment range. This trend is more pronounced in higher values of $\epsilon$ such as 0.5 (50\%) and 0.2 (20\%). This is because a higher value of U provides a larger error range in which to search for minimum cost solutions. For example, for 50\% error guarantee, the error range which is 750 kWh for U = 1500, is three times larger than the error range of 250 kWh for U = 500.

\subsubsection{Scalability Analysis}
In order to perform scalability analysis, we fix the values of $T$: the number of time intervals and $N$: the number of curtailment strategies per node. We vary $M$: the number of nodes for various values of $\epsilon$. As one can see from Figure~\ref{fig:scalability-num-buildings}, the algorithm exhibits a near quadratic increase in the runtime with respect to the number of nodes. This is consistent with the runtime complexity analysis in Section~\ref{sec:proofFPTA1} in the Appendix. 

\begin{figure}[ht]
	\begin{subfigure}{0.48\textwidth}
		\centering
		\includegraphics[width=\linewidth]{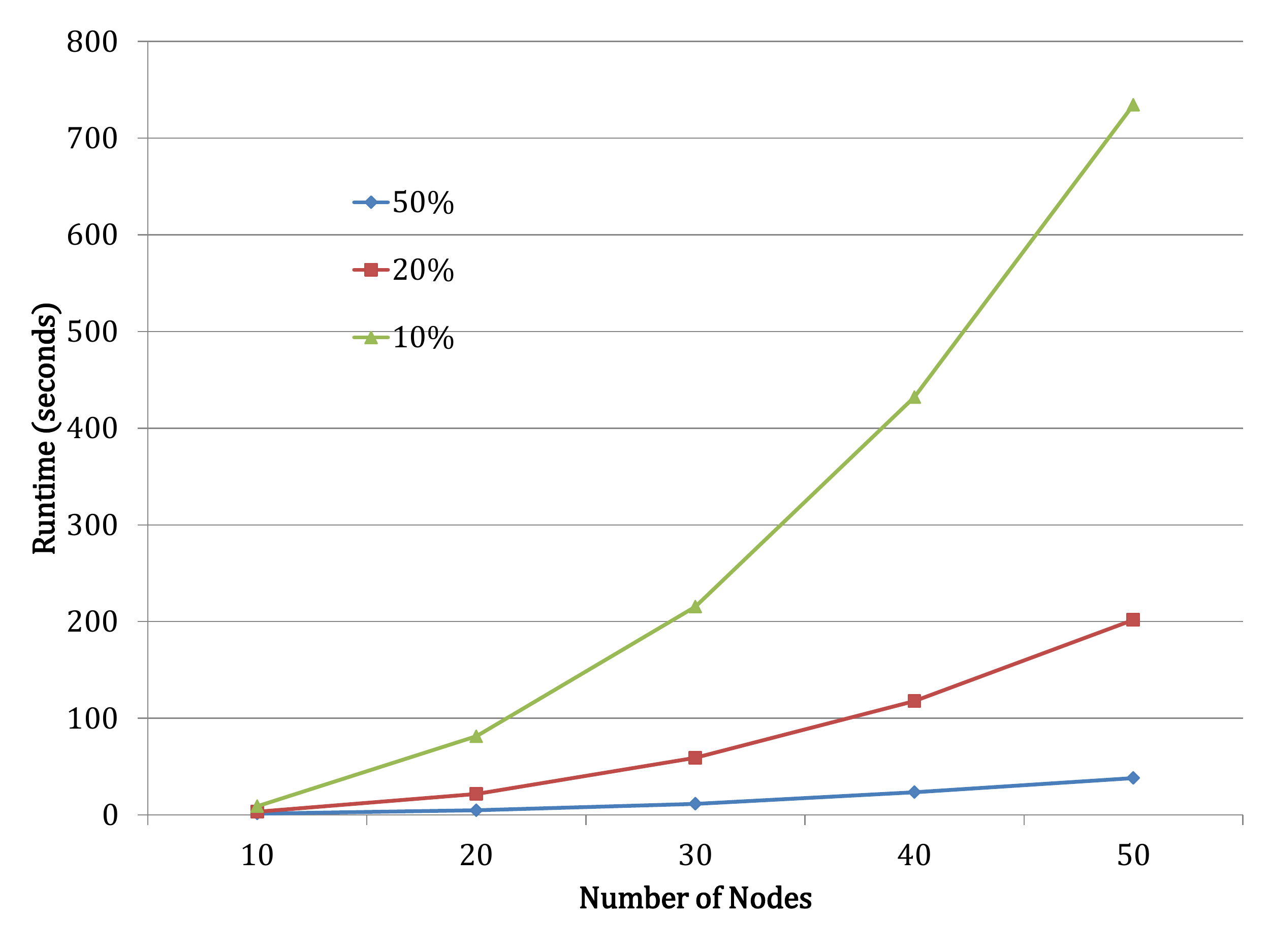}
		\caption{}
		\label{fig:scalability-num-buildings}
	\end{subfigure}
	\begin{subfigure}{0.48\textwidth}
		\centering
		\includegraphics[width=\linewidth]{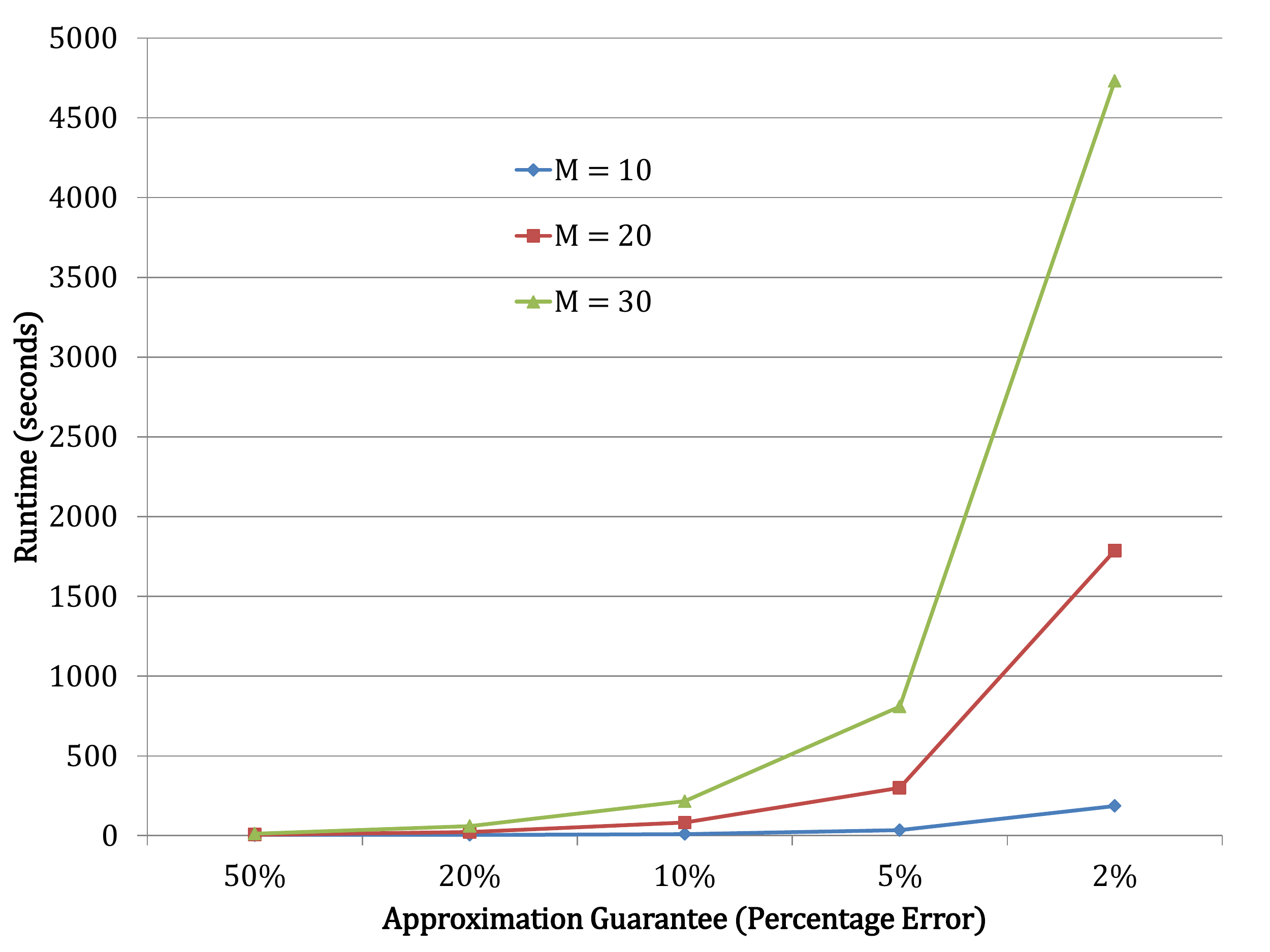}
		\caption{}
		\label{fig:scalability-epsilon}
	\end{subfigure}
	\caption{Runtime of Algorithm 2 versus: (a) Number of nodes for various values of epsilon, (b) $\epsilon$ (denoted as percentage error) for various values of the number of nodes }
\end{figure}

We also analyze the scalability with respect to $\epsilon$ by varying the value of $M$ while keeping $T$ and $N$ fixed. As shown in Figure~\ref{fig:scalability-epsilon}, decreasing $\epsilon$ (increasing accuracy) has a significant impact on runtime. Hence, $\epsilon$ is a parameter that can be used to trade-off accuracy versus computational complexity. 

A reader might comment that the runtimes observed, especially for smaller values of $\epsilon$ are very high. This can be justified as the problem of net-load balancing is NP-hard and hence, high computation capacity is required to increase the accuracy. Our objective in this work is to show that a polynomial time approximation algorithm exists for this NP-hard problem. We did not focus on finding the best optimal solution for the same. Moreover, our experiments are performed on MATLAB. For a real world deployment of this software, using faster programming languages such as C++ will significantly improve the run times (as high as 10-20 times as per the experience of the authors). In the context of real world scenarios, the California ISO's MRTU applications determine the desired generation changes 5-min ahead of the beginning of the interval and the system needs to start moving towards the set point 2.5 minutes ahead of the interval~\cite{makarov2009operational}. Given the typical inverter control delays of the order of milliseconds~\cite{gagrica2015microinverter}, even our naive MATLAB implementation meets the California ISO constraints for 40 nodes with $\epsilon = 0.2$ (20\% error) and 25 nodes for $\epsilon = 0.1$ (10\% error). Similar implementation in C++, assuming a conservative estimate of 10x improvement will meet the constraints for 70 nodes with $\epsilon=0.1$ (10\% error) and 40 nodes with $\epsilon=0.05$ (5\% error).

We also compare our algorithm against demand curtailment selection techniques such as those developed in~\cite{zois2014efficient} and~\cite{gatsis2011cooperative}. We observed that these techniques typically incur errors of around 5-10\% and in the worst case can go as high as 95\%. We excluded the details of this analysis as comparison against the optimal solutions already provides us with an idea of the near optimality of our algorithm.

\subsection{Minimum Cost Net-Load Balancing with Fairness}
\label{ssec:mcfexp}
In Algorithm 3, we introduced the notion of fairness by defining curtailment budget ranges for nodes. Here we evaluate the empirical performance of Algorithm 3. Note that the budget ranges for each node can be set appropriately by user/grid operator. In our experiments, we set the budget $B_b$ for each node $b$ to $\gamma^b_{max}/\sum_{b=1}^{B}\gamma^b_{max} \Gamma$ where $\gamma^b_{max}$ denotes the sum of the maximum curtailment values across all the intervals. We also set $\alpha_b = \alpha$ and vary the value $\alpha$. 

%We evaluate our Minimum Cost Net Load Balancing with Fairness algorithm (Algorithm 3) by varying the (L,U) pairs.
\begin{figure}[ht]
	\centering
	\includegraphics[width=0.9\linewidth]{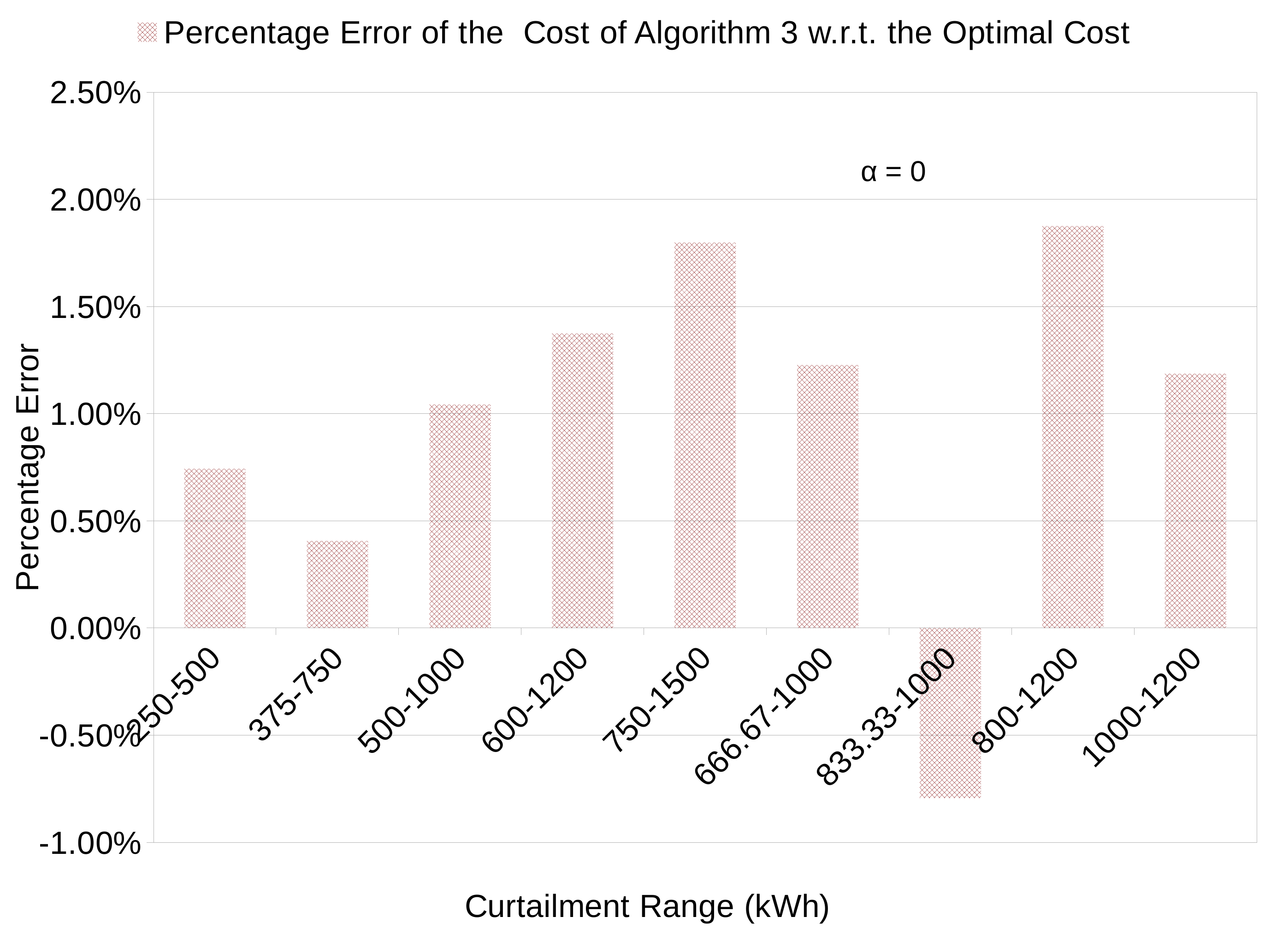}
	\vspace{-0.4cm}
	\caption{Percentage Error of the Cost of Algorithm 3 w.r.t. the Optimal Cost}
	\label{fig:Fairness-accuracy}
\end{figure}

In order to evaluate the accuracy of our algorithm, we compare against the optimal solution obtained from solving the ILP defined in Section~\ref{sssec:fairilp}. Figure~\ref{fig:Fairness-accuracy} shows the percentage error of the cost of the solution produced by algorithm 3 as compared against the optimal cost. Although, the worst case theoretical guarantee is a factor of 4 for quadratic cost function as provided by Theorem 5.3, in practice the algorithm performs much better with errors varying from -0.79\% to 1.88\%. The negative percentage error implies that the cost of the solution from our algorithm was less than the cost of the optimal solution. Note that this is possible because our solution violates certain constraints which the optimal solution does not. 

%However, the complexity of the ILP prevented CPLEX studio to generate results in a reasonable amount of time. Hence, we compared the objective value of the rounded results obtained from our algorithm against the objective value of the LP with fractional solutions. Since, an LP relaxation of an ILP always produces better objective value (albeit with infeasible integral solutions), the results shown will only improve if compared against the optimal ILP solution.  

%\begin{figure}[ht]
%\centering
%\includegraphics[width=0.9\linewidth]{experiments/Budget-overshoot}
%\vspace{-0.5cm}
%\caption{Percentage of Budget Overshoot for 10 Worst Affected Nodes}
%\label{fig:Budget-overshoot}
%\end{figure}

%Figure~\ref{fig:Budget-overshoot} shows 
We also studied the percentage of the budget overshot for the nodes across all the values of L and U. The theoretical guarantee of factor 2 (100\%) provided by Theorem 5.3 is honored in all the cases. We observed that the error is less than 13\% (1.13 factor) for the worst performing node. A similar study for the percentage of the interval curtailment target $\Gamma_t$ undershot for all the intervals reveals that the theoretical guarantee of factor 2 (100\%) provided by Theorem 5.3 is honored in all the cases. In practice the error is less than 7\% (1.07 factor) for the worst performing interval. This implies that the curtailment target for the worst performing interval could not be met and was deficit by 7\%.

%\begin{figure}
%\centering
%\includegraphics[width=0.9\linewidth]{experiments/gamma-undershoot}
%\caption{Percentage Undershoot of Interval Curtailment Target for various intervals (note: the labels of x axis do not denote the actual interval number)}
%\label{fig:gamma-undershoot}
%\end{figure}

%Figure~\ref{fig:gamma-undershoot} shows 

We also calculated the gini coefficient -- which is the most commonly used measure of inequality in economics -- of the curtailment achieved by each node as a proportion of its budget to measure the fairness of curtailment. Figure~\ref{fig:gini} shows the results for various values of $\alpha$ for a curtailment range of 500-1000 with 20 nodes. The value of gini decreases with increasing $\alpha$ as the dispersion of the curtailment decreases. Above $\alpha=0.2$, no feasible solution could be found. 

\begin{figure}
	\centering
	\includegraphics[width=0.7\linewidth]{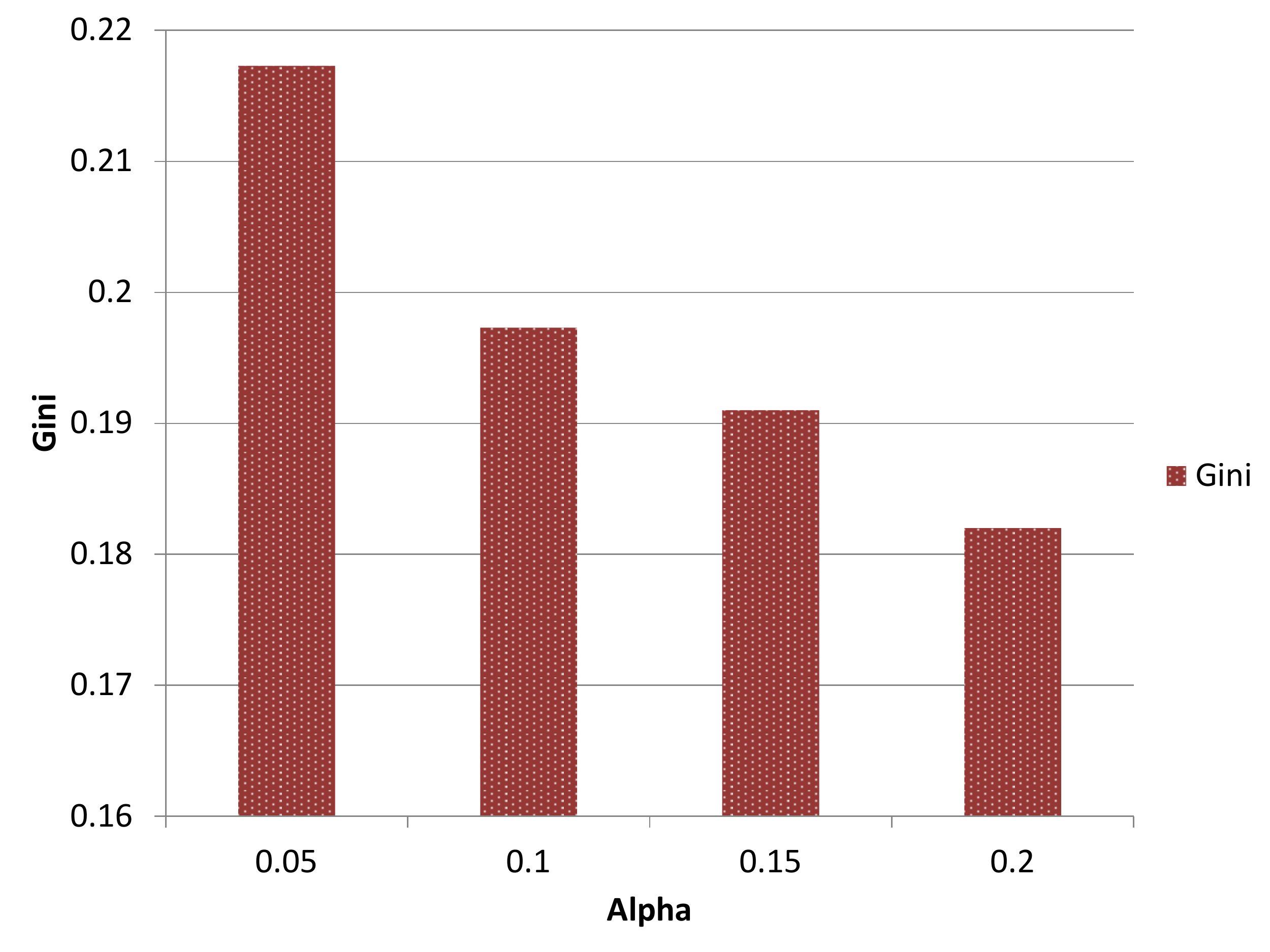}
	\vspace{-0.4cm}
	\caption{Gini Coefficient for various values of Alpha ($\alpha$)}
	\label{fig:gini}
\end{figure}

%Mathematically,
%\begin{flalign}
%G = \frac{\sum_{i=1}^{n}\sum_{j=1}^{n}|x_i - x_j|}{2n\sum_{i=1}^{n}x_i}
%\end{flalign}

\subsection{Online Algorithm for Minimum Cost Net-Load Balancing with Fairness}
In order to evaluate Algorithm 4, we compare the cost of the solutions obtained against the optimal solutions of the ILP defined in Section~\ref{sssec:fairilp} as generated in Section~\ref{ssec:mcfexp}. For various values of L and U pairs, we calculate the percentage error in the cost obtained by the online algorithm (Algorithm 4) with respect to the optimal solutions. The budget values $B_b$ input to Algorithm 4 are same as the ones used in Section~\ref{ssec:mcfexp}. 

\begin{figure}[ht]
\centering
\includegraphics[width=0.9\linewidth]{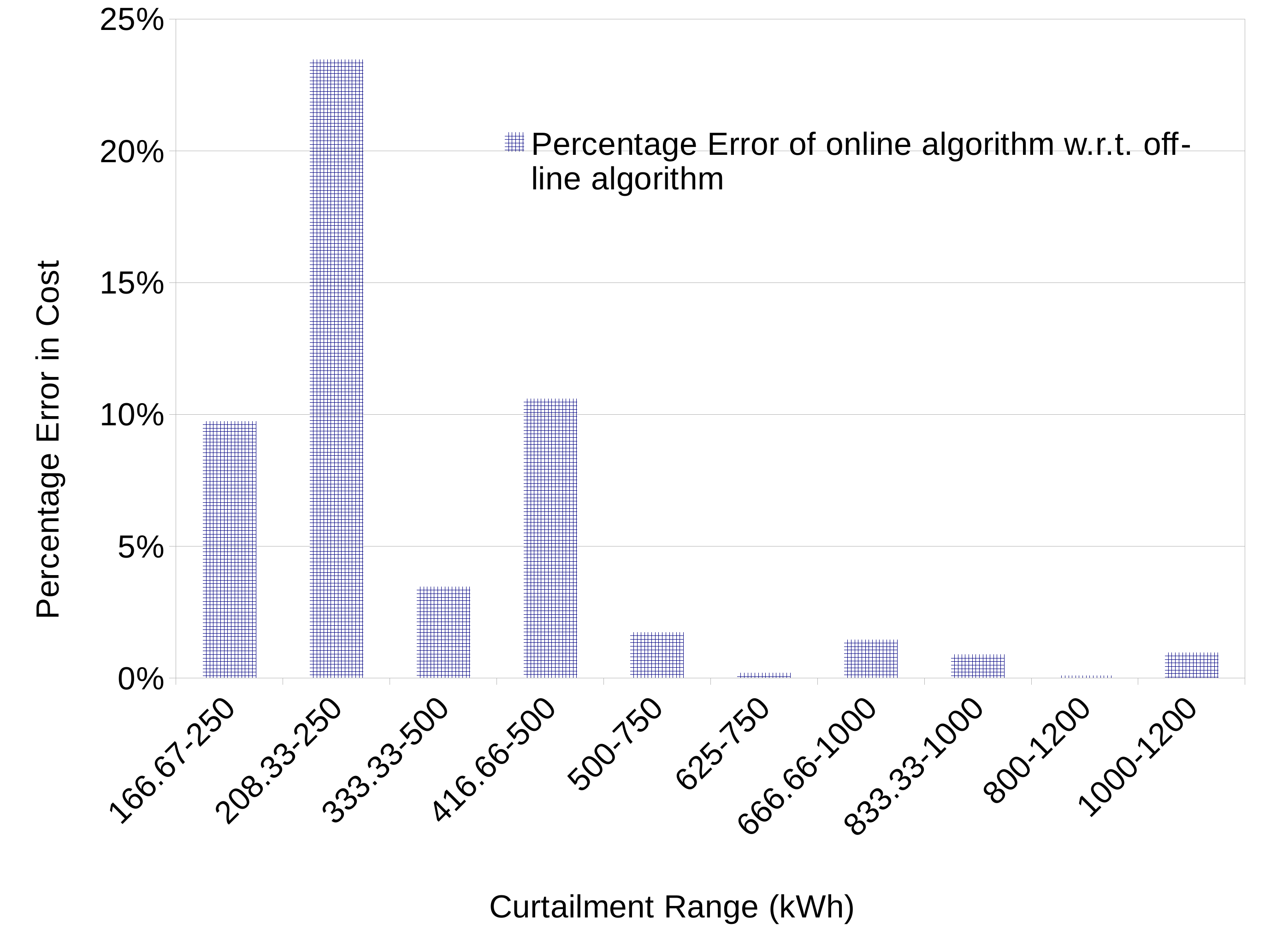}
\vspace{-0.4cm}
\caption{Percentage Error of the Cost of Online Algorithm w.r.t the Optimal Cost}
\label{fig:online-algo}
\end{figure}

Figure~\ref{fig:online-algo} shows the results obtained. Even though we do not provide any guarantee on the worst case bounds for the online algorithm, in practice the error incurred is low. The highest error incurred is around 23\% (factor 1.23). This makes the online algorithm a good candidate for net-load balancing when the predictions for the entire horizon are not known in advance.

\section{Conclusions}
One of the most significant change to the smart grids of future will be the proliferation of PV systems. The current distribution grid is characterized by a few active suppliers such as the utility and a large number of passive consumers with the power flowing unidirectionally from the suppliers to the consumers. However, in future the consumers will become an active participant of the grid with the capability of injecting power into the grid. Net-load balancing under this scenario will be a daunting task.

In this work, we addressed the problem of performing net-load balancing under the assumption that the nodes can be directly (and remotely) controlled by the grid operator. Due to the hardness of this problem, previous works in the literature had to compromise on either computational tractability or accuracy. We showed that it is possible to achieve both the conflicting goals simultaneously. 

However, there are several challenges which still need to be addressed for seamless PV integration. The uncertainty due to the errors in the forecasting algorithms is not considered in this work. Similarly, this work assumes complete observability and direct control of the grid which is true for micro-grids such as industry/university campus but might not be true for a city wide distribution grid. We will focus on addressing the above mentioned challenges in our future works.

\bibliographystyle{ACM-Reference-Format}
\bibliography{buildsys} 
\newpage
\section{Appendix}
\appendix
\section{NP-hardness of Minimum Cost Net-Load Balancing Problem}
\label{sec:nphard}
In order to prove that the problem of minimum cost net-load balancing is NP-hard, we will define a simpler version of the problem and reduce the well known knapsack problem, which is an NP-hard problem to it. Adding any additional constraints to this simpler version will only increase the complexity of the problem.

The simpler version of the problem $\Pi$ is formulated as follows: We are given a set $S$ of node-strategy pairs, where $s_{ij} \in S: i \in \{1,\dots,M\}, j \in \{1,\dots,N\}$ denotes the node-$i$-strategy-$j$ pair, where $M$ is the number of node and $N$ is the number of strategies. Given a curtailment value $\Gamma$, we need to output a $S^* \subseteq S$ such that $\Gamma \leq \sum_{i,j:s_{ij} \in S^*}\gamma(s_{ij})$, where $\gamma(s_{ij})$ denotes the curtailment obtained by $s_{ij}$ and $ \sum_{j=1}^{N}I(s_{ij})\leq 1 \forall i \in \{1,\dots,M\}$, where $I(s_{ij}) = 1$ if $s_{ij} \in S^*$ and 0 otherwise and $\sum_{i,j:s_{ij} \in S^*}\mathcal{C}(s_{ij})$ is minimized, where $\mathcal{C}(s_{ij})$ denotes the cost of curtailment by $s_{ij}$.

\begin{theorem}
	$\Pi$ is NP-hard.
\end{theorem}

\begin{proof}
A 0-1 Knapsack problem~\cite{sahni1975approximate} is defined as follows: Given $M$ elements, with element $i$ having value $v_i$ and size $d_i$, find a sub-set of elements the sum of whose sizes is $\leq D$ and the value is maximized. To reduce this problem into $\Pi$, for each element $i$, we add $s_{i1}$ with $\gamma(s_{i1}) = -d_i$ and $\mathcal{C}(s_{i1}) = -v_i$. We set $\Gamma = -D$. Note that $j \in \{1\}$. One can easily observe that 0-1 knapsack problem has a solution if and only if $\Pi$ has a solution.
\end{proof}

%\section{Computing $C_{min}$ and $C_{max}$}
%\label{sec:cminmax}
%The choice of $C_{min}$ and $C_{max}$ affects the runtime of the algorithm. We can say that all the costs which have the same scaled value fall into the same bucket. Using this definition, $C_{min}$ determines the bucket size and $C_{max}$ determines the number of buckets. The choice of $C_{min}$ and $C_{max}$ should satisfy the following requirement: $0 < C_{min} \leq \text{OPT} \leq C_{max}$, where OPT is the value of optimal solution.
%
%To determine $C_{min}$ and $C_{max}$, we take the cost to curtailment ratio of all the non-zero curtailment values. We multiply the smallest and the largest ratios with $\sum_{t=1}^{T}\Gamma_t$ to get $C_{min}$ and $C_{max}$ respectively. We assume that a non-zero curtailment will incur a non-zero cost. 
\section{Proof of Theorem 5.1}
\label{sec:proofFPTA1}
For a curtailment value $\gamma$, we say that $\widehat{\gamma}$ is the rounded curtailment value. Also, we call $\gamma$ as the unrounded curtailment value.
\textit{Correctness:} We define \textit{Bucket} of $\widehat{\gamma}$ as the range $\mu \widehat{\gamma} - \mu < \gamma \leq \mu \widehat{\gamma}$ i.e. all the curtailment values which get rounded to $\widehat{\gamma}$. $C^{\Theta}_{\widehat{\gamma}}$ denotes the cost assigned to the bucket in table $\Theta$. For each interval $t$, using induction on Equation~\ref{eqn:dp}, it is easy to show that $C^{\Theta_t}_{\widehat{\gamma}} \;\forall t, 0 \leq \widehat{\gamma} \leq \widehat{\Gamma}$ will be the minimum cost required to achieve any curtailment value in the Bucket of $\widehat{\gamma}$ for table $\Theta_t$. Hence, for each element $e_t = (c_t, \widehat{\gamma}_t)$, $c_t$ is the minimum possible cost to achieve any curtailment value in the Bucket of $\widehat{\gamma}_t$ corresponding to the table $\Theta_t$. Now, if we can show that the range of curtailment values covered by the elements in $S_t \forall t$ contains the curtailment value chosen by any optimal solution, and that the cumulative curtailment value of any optimal solution at any time $t$ is contained in the entries $\Phi(:,t)$ (i.e. entries corresponding to time interval $t$), then using induction on $\Phi$ table, we can show the aggregate cost of the solutions obtained by Algorithm 2 will be less than or equal to the optimal solution. 
Now, in each interval, the lowest rounded curtailment value $\widehat{\gamma}$ considered to create the set $S_t$ is $\widehat{\gamma}_t$. Hence, the lowest unrounded curtailment $\gamma$ satisfies: $\gamma_t - \mu < \gamma \leq \gamma_t$ i.e. $\gamma \leq \gamma_t$. Now, the maximum rounded curtailment value $\widehat{\gamma}$ considered by Algorithm 2 in the $\Phi$ table is $\widehat{\Gamma}$. So, the maximum curtailment value $\gamma$ considered in the corresponding bucket is $\mu \widehat{\Gamma} \geq \Gamma$.  Hence, the range of curtailment values considered by Algorithm 2 covers the range in which optimal solution can reside and so Algorithm 2 does not miss any possible solution of a lower cost. 

\textit{Runtime:} The Algorithm fills $T$ $\Theta$ tables each of which is of size $M \widehat{\Gamma}$. We assume that $\Gamma = O(T\Gamma_{min})$.  Each entry of $\Theta$ requires $O(N)$ time. Hence, the total runtime for line~\ref{line:filltheta} and~\ref{line:computest} of Algorithm 2 is $O(\frac{M^2 N T^2}{\epsilon})$.

The number of entries in table $\Phi$ is $O(T \widehat{\Gamma})$. Each entry requires $O(\widehat{\Gamma})$ time. Hence, the total required time to fill $\Phi$ is $O(T \widehat{\Gamma}^2 = O(\frac{T^3 M^2}{\epsilon^2}))$. Once the tables are filled, the for loop requires $O(TMN \widehat{\Gamma})$ to output the strategies. Hence, the algorithm is polynomial in the input size $M, N, T$ and $\frac{1}{\epsilon}$. 
  
\textit{Approximation Guarantee:} Let $\widehat{\Gamma}_x = \sum_{i=1}^{M}\sum_{t=1}^{T}\widehat{\gamma}_{it}$ be the solution from our algorithm, where $\widehat{\gamma}_{it}$ denotes the rounded curtailment by node $i$ in time $t$. From line~\ref{line:computest} of the algorithm, for all $t$, we know that $\sum_{i=1}^{M}\widehat{\gamma}_{it} \geq \widehat{\Gamma}_t$. Also, $\frac{\gamma_{it}}{\mu} \leq \widehat{\gamma}_{it} \leq \frac{\gamma_{it}}{\mu} + 1$ by the definition of $\widehat{\gamma}_{it}$. This implies that $\frac{\Gamma_t}{\mu} \leq \sum_{i=1}^{M}(\frac{\gamma_{it}}{\mu} + 1) \leq \sum_{i=1}^{M}\frac{\gamma_{it}}{\mu} + M$. So, $\Gamma_t - \mu M \leq \sum_{i=1}^{M}\gamma_{it} \implies \sum_{i=1}^{M}\gamma_{it} \geq \Gamma_t - \epsilon \Gamma_{min} \geq (1 - \epsilon) \Gamma_t$ as $\Gamma_t \geq \Gamma_{min}$. Hence, in each interval, the curtailment target constraint (Equation~\ref{eqn:achievegamma}) is violated by a maximum factor of $(1 - \epsilon)$.

Now, from line~\ref{line:gammaless} of the algorithm, $\widehat{\Gamma}_x \leq \widehat{\Gamma}$. So, $\frac{\Gamma_x}{\mu} \leq \frac{\Gamma}{\mu} + 1$. This means that $\Gamma_x \leq \Gamma + \epsilon\frac{\Gamma_{min}}{M} \leq \Gamma (1+\frac{\epsilon}{M}) \leq \Gamma(1+\epsilon)$ as $\Gamma_{min} \leq \Gamma$ and $M \geq 1$. Hence, the aggregate curtailment target constraint (Equation~\ref{eqn:gammaless}) is violated by a maximum factor of $(1+\epsilon)$.

\section{Proofs for Theorems 5.2-5.3}
\label{sec:prooflpround}
Let $\gamma'_{bt} = \sum_{j=1}^{N}\gamma_{bj}(t)x_{bj}(t)$ be the curtailment value obtained for node $b$ in time interval $t$. Let $\gamma^i_{bt}$ and $\gamma^{i+1}_{bt}$ be the curtailment values of strategies between which $\gamma'_{bt}$ falls i.e. $\gamma^i_{bt} \leq \gamma'_{bt} \leq \gamma^{i+1}_{bt}$. Now, if $\gamma'_{bt}$ is rounded up to $\gamma^{i+1}_{bt}$, this implies $\gamma'_{bt} \geq (\gamma^i_{bt}+\gamma^{i+1}_{bt})/2 \geq \gamma^{i+1}_{bt}/2$. Summing over all values of $b, t$ ensures that Equation~\ref{eqn:mcnlfgammaless} and the upper bound of Equation~\ref{eqn:mcnlffair} are violated by at most a factor of 2. Similarly it can be shown that if the objective function is linear, it will be bounded by a factor of 2 and if it is quadratic, it will be bounded by a factor of 4.

Now, in order to provide a bound on the constraint violation of Equation~\ref{eqn:mcnlfachievegamma} and the lower bound of Equation~\ref{eqn:mcnlffair}, we make an assumption that $\gamma^{i+1}_{bt} \leq (2k-1)\gamma^{i}_{bt}$. Using this assumption, we can ensure that in the worst case $\gamma^{i}_{bt} \geq \gamma'_{bt}/k$ i.e. the constraint is violated at most by a factor of $k$. Note that if $\gamma^i_{bj} = 0$, we cannot provide any guarantee.

The analysis above can be used by the grid operator to improve the efficiency of net-load balancing. The grid operator can setup curtailment configurations (e.g. load curtailment strategies such as GTR, etc.) such that the difference between two consecutive curtailment values is not vary large. However, it is not always possible to control the curtailment configurations. Hence, the grid operator can develop techniques by noticing that a higher curtailment value can be expected to have less violations for two reasons: (1) very few values of $\gamma'_{bt}$ will have $\gamma^i_{bj} = 0$ and (2) since a large number of non-zero curtailment strategies will be selected, several of them will be rounded up thus reducing the possibility that Equation~\ref{eqn:mcnlfgammaless} is violated by a large factor. Hence, if the required curtailment value is small for a curtailment horizon, the grid operator can reduce the number of nodes which participate in the curtailment horizon.

\end{document}